\newif\ifconf
\conftrue

\newif\ifcomm
\commtrue

\newif\ifblind
\blindfalse

\newif\ifs
\sfalse

\newif\ifhyperref
\hyperreftrue

\newif\ifacm
\acmfalse

\newif\ifusenix
\ifacm
   \usenixfalse %
\else
    \usenixfalse %
\fi

\newif\ifhotnets
\ifacm
   \hotnetsfalse %
\else
    \ifusenix
        \hotnetsfalse %
    \else
        \hotnetsfalse %
    \fi
\fi
    
\newif\ifacmart
\ifacm   %
   \acmarttrue
\else
   \acmartfalse
\fi

\newif\ifieee
\ifacm
   \ieeefalse
\else
    \ifusenix
        \ieeefalse
    \else
        \ifhotnets
            \ieeefalse
        \else
            \ieeetrue
        \fi
    \fi
\fi

\ifconf
    \newcommand{\Conf}[1]{#1}
    \newcommand{\TR}[1]{}
    \newcommand{\Journal}[1]{}  %
    \newcommand{\OnlyTR}[1]{}   %
\else
    \newcommand{\Conf}[1]{}
    \newcommand{\TR}[1]{#1}
    \newcommand{\Journal}[1]{#1}  %
    \newcommand{\OnlyTR}[1]{}   %
\fi

\ifacm
	\ifacmart
	    \ifblind
            \documentclass[sigconf,anonymous,nonacm]{acmart}
        \else
            \documentclass[sigconf]{acmart}
        \fi
    \else
      \Conf{
 	     \documentclass[letterpaper,conference]{sig-alternate-10pt}
      }
      \TR{
     	 \documentclass[letterpaper,conference]{sig-alternate}
      }
      \usepackage[
       pass,%
      ]{geometry}
	\fi
\else
	\ifusenix  %
		\documentclass[letterpaper,twocolumn,10pt]{article}
        \usepackage{usenix-2020-09}
        \usepackage{filecontents}
    \else %
	    \ifhotnets  %
            \documentclass{hotnets23}
    	\else  %
            \Conf{
                  \documentclass[conference]{IEEEtran}
            }
            \TR{
                  \documentclass[journal]{IEEEtran}
            }
        \fi
    \fi
\fi

\ifieee %
    \usepackage[noadjust]{cite} %
	\usepackage[cmex10]{amsmath}
	\usepackage{amsthm}  %
    \newtheoremstyle{boldthm}{}{}{\itshape}{}{\bfseries}{.}{ }{\thmname{#1}\thmnumber{ #2}\thmnote{ (#3)}} %
	\theoremstyle{boldthm}
\fi

\ifacmart 
\else 
    \ifhotnets
        \usepackage{times}  
        \ifhyperref
            \usepackage[hidelinks]{hyperref}
            \hypersetup{pdfstartview=FitH,pdfpagelayout=SinglePage}
        \fi

    \else
        \ifhyperref
            \usepackage[hidelinks]{hyperref} %
            \hypersetup{pdfstartview=FitH,pdfpagelayout=SinglePage}
        \fi
    \fi
\fi

\ifacmart
\else
  \newtheorem{theorem}{Theorem}%
  \newtheorem{property}{Property} %

\fi

\ifieee %
    \newcommand{\bp}{\begin{IEEEproof}}     %
    \newcommand{\bpo}{ \begin{IEEEproof}[Proof Outline] }
    \newcommand{\ep}{\end{IEEEproof}}       %
    \newcommand{\proofof}[1]{\begin{IEEEproof}[Proof of #1]} %
\else
    \newcommand{\bp}{\begin{proof}}
    \newcommand{\bpo}{ \begin{proof}[Proof Outline] }
    \newcommand{\ep}{\end{proof}}       %
    \newcommand{\proofof}[1]{\begin{proof}[Proof of #1]} %
\fi

\ifieee
    \usepackage[justification=justified, font=footnotesize]{caption}
    \usepackage[labelformat=simple, font=footnotesize]{subcaption}  
    \DeclareCaptionLabelSeparator{periodspace}{.\quad}
    \ifconf
    \else
    \fi
\else %
    \ifs %
    	\usepackage[small]{caption} 
    \else
    	\usepackage{caption}
    \fi
    \usepackage[labelformat=simple]{subcaption}
\fi

\usepackage{comment}
\usepackage{graphicx}
\usepackage{amssymb}
\usepackage{url}
\usepackage[normalem]{ulem} %
\usepackage{xcolor} %

\usepackage{xspace}
\usepackage{balance}
\usepackage{booktabs}  %
\usepackage{tabularx}  %

\usepackage{array}
	\newcolumntype{C}[1]{>{\centering\let\newline\\\arraybackslash\hspace{0pt}}m{#1}}
\usepackage[inline,shortlabels]{enumitem} 
\usepackage{algorithm}
\usepackage[noend]{algorithmic}

\usepackage[capitalize,noabbrev]{cleveref}
\crefname{equation}{Eq.}{Eqs.}
\Crefname{equation}{Eq.}{Eqs.}
\crefname{figure}{Fig.}{Figs.}
\Crefname{figure}{Fig.}{Figs.}
\crefname{table}{Tab.}{Tabs.}
\Crefname{table}{Tab.}{Tabs.}
\crefname{property}{Property}{Properties}
\Crefname{property}{Property}{Properties}
\crefname{section}{\S}{Secs.}
\Crefname{section}{\S}{Secs.}
\crefname{subsection}{\S}{Secs.}
\Crefname{subsection}{\S}{Secs.}
\crefformat{subsection}{\S#2#1#3}

\ifcomm
	\newcommand{\mycomm}[3]{{\footnotesize{{\color{#2} \textbf{[#1: #3]}}}}}
     \newcommand{\Fmycomm}[3]{\footnote{{{\color{#2} \textbf{[#1: #3]}}} }}
\else
    \newcommand{\mycomm}[3]{}
    \newcommand{\Fmycomm}[3]{}
\fi

\newcommand{\new}[1]{\textcolor{black}{#1}}
\newcommand{\camera}[1]{#1}

\newenvironment{red}{\color{black}}{}

\ifs
	\usepackage[compact]{titlesec} %
		
	\newcommand{\ReduceVSpace}{\vspace{-13pt}}

    \usepackage[inline]{enumitem}
        \setlist{leftmargin=*} %
    \newcommand{\T}[1]{\par\vspace{2pt plus 1pt minus 1pt}\noindent\textbf{#1}} %

\else

	\newcommand{\ReduceVSpace}{\vspace{-3pt}}

    \usepackage[inline]{enumitem}
    \newcommand{\T}[1]{\par\smallskip\noindent\textbf{#1}} %
\fi

\newcommand{\be}{\begin{equation}}
\newcommand{\ee}{\end{equation}}

\newcommand{\vx}{\checkmark\kern-1.1ex\raisebox{.7ex}{\rotatebox[origin=c]{125}{--}}} %

\newcommand{\OPT}{\textsc{opt}}

\newcommand{\name}{\textsc{ToggleCCI}\xspace}
\newcommand{\AlwaysCCI}{\textsc{Always-CCI}\xspace}
\newcommand{\AlwaysVPN}{\textsc{Always-VPN}\xspace}
\newcommand{\AvgAll}{\textsc{Avg(All)}\xspace}
\newcommand{\AvgMonth}{\textsc{Avg(Month)}\xspace}

\newcommand{\demand}{d^{p,t}}
\newcommand{\costVPN}{c_{\textsc{VPN}}^{p,t}}
\newcommand{\costCCI}{c_{\textsc{CCI}}^{p}}
\newcommand{\leaseCCI}{L_{\textsc{CCI}}}
\newcommand{\leaseVlan}{V_{\textsc{CCI}}^{p}}
\newcommand{\leaseVPN}{L_{\textsc{VPN}}^{p}}
\newcommand{\Delay}{D\xspace}

\newcommand{\contract}{T_{\textsc{CCI}}}

\newcommand{\History}{h}

\newcommand{\aggVPN}{R_{\textsc{VPN}}}

\newcommand{\aggCCI}{R_{\textsc{CCI}}}

\newcommand{\thresholdEnter}{\theta_{1}}

\newcommand{\thresholdStay}{\theta_{2}}
\newcommand{\tstate}{T_{\mathrm{state}}}

\newcommand\blfootnote[1]{%
  \begingroup
  \renewcommand\thefootnote{\fnsymbol{footnote}}%
  \renewcommand\thempfootnote{\fnsymbol{mpfootnote}}%
  \footnote[0]{#1}%
  \endgroup
}

\newcommand{\newVar}[2]{\newcommand{#1}{\ensuremath{#2}\xspace}}
  \newVar{\server}{S}
  \newVar{\client}{C}
  \newVar{\rclient}{R_c}
  \newVar{\rserver}{R_s}
  
\providecommand{\vs}{{vs.}\xspace}
\providecommand{\ie}{{i.e.,}\xspace}
\providecommand{\eg}{{e.g.,}\xspace}

\usepackage{seqsplit}
\renewcommand{\seqinsert}{\ifmmode\allowbreak\else\textbackslash\fi}
\newcommand{\cmd}[1]{
\renewcommand{\\}{\textbackslash\newline\mbox{}\hfill}
\par\noindent\rule{0pt}{1.2\baselineskip}%
{\color{gray!50!black}\texttt{#1}}%
\vspace{.2\baselineskip}%
\newline%
}

\begin{document}

\title{Understanding Cross-Cloud Interconnects: \\ Hands-On Measurements and Cost Optimization}

\ifblind 
	\ifacmart
    \else
    	\author{}%
	\fi
\else %
  \ifacm %
	\ifacmart
        \ifhyperref
            \newcommand{\aut}[2]{#1\texorpdfstring{$^{#2}$}{(#2)}}  %
        \else
            \newcommand{\aut}[2]{#1$^{#2}$}
        \fi
        \author{
          \aut{Eitan Eliav}{1},
          \aut{Isaac Keslassy}{1},
          \aut{David Breitgand}{2},
          \aut{Dean H. Lorenz}{2},
          \aut{Avi Weit}{2}
        }%
              \affiliation{
        $^1$ \textit{Technion} \quad 
        $^2$ \textit{IBM Research -- Haifa} }
          \renewcommand{\shortauthors}{F. Author \textit{et al.}}     %
 \acmConference[CoNEXT'23]{ACM CoNEXT}{December 2023}{Seoul, South Korea}
\settopmatter{printfolios=true,printacmref=false} %
   	  \setcopyright{none}
	\else %
    	\author{List of authors}
     \fi

  \else %
  \ifusenix %
      \author{
      {\rm Your N.\ Here}\\
      Your Institution
      \and
      {\rm Second Name}\\
      Second Institution
      } %
      \else %
      \ifhotnets %
            \author{TBD}
      
            \else %
               \Conf{
                    \ifhyperref
                        \newcommand{\aut}[2]{#1\texorpdfstring{$^{#2}$}{(#2)}}  %
                    \else
                        \newcommand{\aut}[2]{#1$^{#2}$}
                    \fi
                  \author{
                     \IEEEauthorblockN{\large   
                        \aut{Eitan Eliav}{1},
                        \aut{Isaac Keslassy}{1,2},
                        \aut{David Breitgand}{3},
                        \aut{Dean H. Lorenz}{3},
                        \aut{Avi Weit}{3}
                          }
                    \IEEEauthorblockA{
                              $^1$ \textit{Technion} \quad 
                              $^2$ \textit{UC Berkeley} \quad 
                              $^3$ \textit{IBM Research -- Haifa}}\\
                    }
                }
                \TR{\author{Grad Student and Isaac~Keslassy,~\IEEEmembership{Senior Member,~IEEE} \thanks{\textbf{[Add this 1st paragraph in 1st journal submission only:]} This paper was presented in part at IEEE Infocom '23, New York, May 2023. Additions to the conference version include new theorems, complete proofs that were previously omitted for space reasons, and additional simulation results.
        
                     G. Student and I. Keslassy are with the Department of Electrical and Computer Engineering, Technion, Israel (e-mails: \{grad@tx., isaac@\}.technion.ac.il).
                }}}
        \fi %
    \fi %
  \fi %
\fi %

\OnlyTR{\markboth{Technical Report TR16-01, Technion, Israel}{}}
\Journal{\markboth{Journal of \LaTeX\ Class Files,~Vol.~14, No.~8, August~2015}{}}%

\ifacmart
\else
	\maketitle
        \blfootnote{© 2026 IEEE. Personal use of this material is permitted. See IEEE Intellectual Property Rights for details.}
\fi

\ifacm %
    \sloppypar
\else 
    \ifhotnets
        \sloppypar
    \else
    \fi
\fi

\begin{abstract}
New services such as Google Cross-Cloud Interconnect (CCI) address the rise in fast and large-scale cross-cloud data transfers. CCI offers dedicated high-throughput links with low per-GB transfer costs, but also involves high fixed leasing fees and multi-day provisioning delays. This combination makes cost optimization difficult because traffic patterns are unpredictable. 

This paper presents the first comprehensive study of CCI‑like services. We begin with an empirical characterization of CCI and its alternatives using direct measurements across AWS–GCP interconnects. We then introduce \name, a new dynamic cost‑optimization algorithm designed to handle provisioning delays and uncertainty in future demand. \name adapts by switching between VPN and CCI based on cost trends observed over a sliding time window. We prove that \name achieves asymptotic optimality under sustained high‑demand or low‑demand regimes. Finally, using real‑world traffic traces, we show that \name consistently tracks the best static policy for each scenario and delivers substantial cost savings.

\end{abstract}


\section{Introduction}
\label{sec:introduction}

Organizations increasingly deploy massive and complex applications across multiple cloud providers to improve availability, reduce vendor lock-in, and optimize cost and performance. This multi-cloud trend brings new technical challenges, particularly in the area of cross-cloud connectivity. Moving data between clouds is not only a performance concern but also a major contributor to operational costs~\cite{SkyPlane-2023, Matan-Avneri-2025,SkyPilot-2023,lucidity2024multicloud,seemoredata2024snowflake,Contrail}. %

The simplest way to transfer data across clouds is via the Public Internet. Although this approach is flexible and easy to set up, it suffers from variable performance, limited security guarantees, and potential privacy risks.  
To address these privacy concerns, cloud providers offer private networking options based on VPNs~\cite{aws-vpc,GCP-VPC,IBM-VPC,Azure-vnet}. However, this method can become costly for users who transfer large volumes of data.

Another approach is to use a dedicated physical connection that enables users to establish a direct link across clouds. A notable recent advancement in this space is Google’s \textit{Cross-Cloud Interconnect (CCI)}~\cite{CCI,googleCCI2024,googleCCI2025}. 
\begin{red} CCI enables users to manage cross-cloud connectivity directly through the GCP Console, with a native Google Cloud API and typical 99.99\% availability SLA guarantees. 
\end{red}
\new{CCI can be seen as} an evolution of the hybrid cloud, in which the user's premises are connected to the cloud via a private dedicated link bypassing the Public Internet~\cite{azure_expressroute,aws_direct_connect,gcp_dedicated_interconnect2024}.  
CCI is the first cloud native solution for direct collocated cross-cloud connectivity that automates the process of establishing communication links. Competing solutions either require manual efforts or rely on intermediary partners to establish connectivity. Azure also has recently provided another collocated solution with native connectivity to Oracle, but it does not include other clouds yet~\cite{oracle1,oracle2}.
CCI is also an alternative to third-party vendors \new{such as Megaport and Aviatrix} who offer physical and virtual private connectivity solutions for cross-cloud scenarios~\cite{Aviatrix,megaport2025,Fortinet}.

In this paper, we focus on exploring the tradeoffs and challenges involved with the cloud-native CCI scenario\new{, as there is no available academic work or objective evaluation about CCIs. In particular, we seek to explore two core metrics: \\
(1)~\textbf{Performance:} How does CCI compare to VPN and the Public Internet? Does it indeed provide its advertised bandwidth? Do factors like the number of connections or the cloud region impact its performance? \\
(2)~\textbf{Cost:} The cost tradeoff between CCI and VPN is hard to model, because of three main challenges. (a)~\textit{Time lag}: unlike VPNs, which are pay-as-you-go and can be provisioned instantly, CCI requires advance provisioning that typically takes several business days. Therefore the user essentially needs to bet on future demand.}
\begin{red}
(b)~\textit{Complex pricing structure:} Cloud provider pricing is often complex and varies significantly across regions and services, especially in multi-cloud setups, leading users to rely on specialized calculators to estimate costs~\cite{plewnia2023cross, aoshima2020predesign}. CCI further increases this complexity because it is a shared organizational resource: the organization pays a base leasing fee for the interconnect, and each connection incurs an additional charge.
So the user needs to consider all of its applications at once to decide whether to lease a CCI link.
(c)~\textit{Different per-GB pricing mechanisms:} Comparing the CCI pricing model with VPN is non-trivial, as VPN pricing is based on tiered egress pricing, where the per-GB cost decreases with higher usage, whereas CCI has a fixed per-GB cost.
\end{red}

\textbf{Contributions.}
\new{We start by exploring the \textit{performance} of CCI with} the first hands-on measurement study of CCI. %
Our experiments reveal practical behaviors that are not captured by provider documentation. \begin{red}
We find that CCI links strictly enforce their nominal bandwidth at the link level.
Nevertheless, individual users sharing the same CCI link may sometimes observe
throughput exceeding the capacity they provisioned.
\end{red}

\new{Second, we focus on the \textit{cost} tradeoff between CCI and VPN. While performance can be directly measured, cost is more complex, as it depends on the {algorithm} that the user adopts to decide on whether to use CCI or VPN. We prove that there is no optimal online algorithm, and in fact that no algorithm can even guarantee a fixed competitive ratio that is independent of the problem settings. Instead, we argue that in practice we can use past demand to model future demand, and introduce a heuristic } 
online algorithm, \textit{\name}, which dynamically switches between VPN and CCI based on recent demands. \name is designed to deal with the complex provisioning delays and leasing commitments of CCI, along with the complex pricing structure of CCI and VPN that depends on both time and volume.
We prove that \name achieves near-optimal cost in both low and high demand scenarios. 

\begin{red}
Finally, we evaluate \name on real-world workloads, MIRAGE and Puffer, as well as synthetic workloads, and show that it consistently outperforms baseline strategies. For example, with the MIRAGE workload, at the breakeven traffic rates where the total VPN and CCI costs are equal, \name reduces the total cost by an average factor of $1.8\times$. 
We evaluate many configurations to analyze the algorithm sensitivity: transfers from GCP to AWS and in the reverse direction; from GCP to Azure and back; both in single-continent and in multi-continent scenarios. Across all configurations, the results consistently show that \name either outperforms alternative methods or approaches the best performance. \end{red}

\camera{The source code is publicly available~\cite{github}.
}

\section{Related Work}\label{sec:related}

While there is a rich body of literature on cross-cloud networking and measurement, none of the prior works study the performance or cost optimization of CCI. 

For example, recent efforts such as~\cite{Kandregula2022Evaluating,8538558} evaluate multi-cloud networking performance, and CloudCast~\cite{Itzhak2022CloudCast} presents a system for monitoring and comparing cross-cloud latency and throughput across AWS, GCP, and Azure. Others focus more narrowly, \eg on characterizing throughput dynamics within GCP~\cite{GCP_measurments}.
However, to the best of our knowledge, no prior work conducts hands-on measurements or develops cost-aware frameworks for CCI usage. 

Similarly, there is extensive research on multi-cloud task scheduling and resource management~\cite{tang2021reliability,zhu2021task}. \begin{red} Another work, called Paraglider~\cite{paraglider}, provides an abstraction layer for multi-cloud networking, enabling cloud-agnostic specification of networking and simplifying deployment across multiple cloud providers. However, it does not consider privacy or cost optimization. \end{red} In addition, {Skyplane}~\cite{SkyPlane-2023} formulates the cost-throughput trade-off as a Linear Programming (LP) problem. It dynamically plans and routes cross-cloud data transfers in real-time, achieving substantial speedups while controlling egress costs. However, Skyplane relies solely on Public Internet paths and does not incorporate any privacy guarantees. CloudPilot~\cite{cloudpilot} also looks at optimizing the cost of cloud-proxy placement. Finally, many commercial systems also aim to reduce multi-cloud expenses~\cite{spot_netapp, cast_ai, cloudhealth_vmware}. They offer integrated solutions that optimize total cloud billing. 
Still, none of these studies and systems consider the unique characteristics or pricing constraints of CCI links to optimize a CCI \vs VPN tradeoff, rendering them inapplicable in our setting.

\section{Background}\label{sec:background}

\T{Google Cross-Cloud Interconnect (CCI).}  
CCI is a service that offers high-bandwidth, dedicated connectivity between Google Cloud and other cloud providers~\cite{CCI}. This connectivity is established through designated \textit{colocation facilities}, each tied to a specific external provider and geographically restricted, typically within a single continent. CCI is a shared resource, allowing multiple connections by attaching separate VLAN attachments. Deployments are available in 10~Gbps or 100~Gbps capacities and are backed by Google’s 99.99\% availability SLA, with no additional hardware required.

Establishing a CCI connection is not immediate: the process includes email correspondence with Google’s service center and may take a few business days to complete~\cite{gcp-cci-aws-establish}. %
Once established, the CCI link is billed at a fixed hourly rate, regardless of the volume of traffic transmitted. In addition, there is a fixed price per GiB of data egress from GCP.

\T{Virtual Private Cloud (VPC)} is a logically isolated virtual network within a public cloud~\cite{aws-vpc,GCP-VPC,IBM-VPC,Azure-vnet}. Users deploy all workloads, such as virtual machines (VMs) and Kubernetes (K8s) clusters, inside VPCs, which are not accessible to other users unless explicit permission is granted.

\T{VLAN attachments}~\cite{gcp-VLAN} provide the logical link that connects a specific VPC to a CCI within the same geographic area. Each VPC that wants to use an existing CCI must lease its own VLAN attachment. These attachments are available in a range of bandwidth options and incur an hourly charge based on the selected capacity~\cite{CCI_pricing_catalog}. Unlike the CCI provisioning process, VLAN attachments are fully automated and become active immediately upon request. At least one VLAN attachment is required for an active CCI connection, and multiple attachments can be associated with the same CCI.

\T{Establishing a CCI link} involves several steps:  
(1)~Lease a physical port from both Google and another cloud provider at the same colocation facility;  
(2)~Create a virtual connection from VPCs to the physical port (a VLAN attachment in GCP and a Virtual Interface (VIF) in AWS);  
(3)~Establish routing policies and BGP sessions to enable inter-cloud traffic.  
Full technical details of the setup process are available in~\cite{gcp-cci-aws-establish}.

\T{Virtual Private Network (VPN)} is a service supplied by all major cloud providers~\cite{azure_vpn,ibm_cloud_vpn,aws-vpn-overview,gcp-vpn-overview}. 
VPNs offer high flexibility and are easy to deploy, making them a default choice for many multi-cloud scenarios. Cost-wise, VPNs typically involve a fixed hourly charge for maintaining the connection and variable fees for data transferred out of the network.%

\newcommand{\figtestbed}{\begin{figure}[t]
    \centering
    \def\sscale{.27}
    \includegraphics[scale=\sscale,trim=.85cm 1.4cm .25cm .5cm,clip]{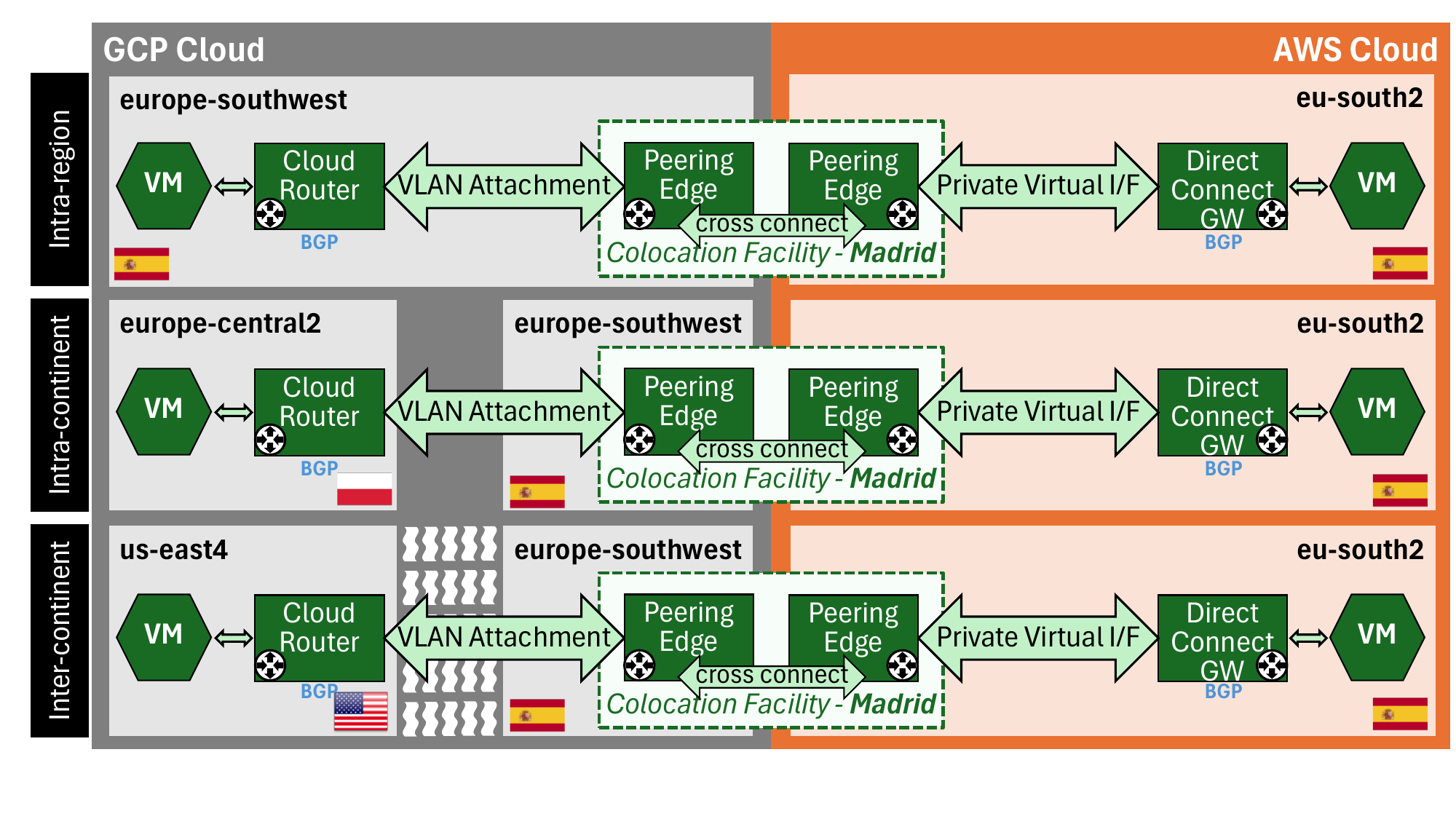}
    \caption{Measurements testbed: measuring network performance across AWS and GCP.    
    }\label{fig:testbed} 
    \ReduceVSpace
\end{figure}}
\newcommand{\figvpn}{\begin{figure*}
\def\sscale{.38}
\centering
\captionsetup[subfigure]{font=footnotesize,skip=4pt}
\begin{subfigure}[t]{.32\linewidth}%
    \includegraphics[page=2,scale=\sscale,trim=2cm 8.8cm 11.5cm 3cm,clip]{figures/vpn_pub.pdf}
    \caption{AWS Spain $\rightarrow$ GCP Spain}
\end{subfigure}
\hfill
\begin{subfigure}[t]{.32\linewidth}%
    \includegraphics[page=1,scale=\sscale,trim=2cm 8.8cm 11.5cm 3cm,clip]{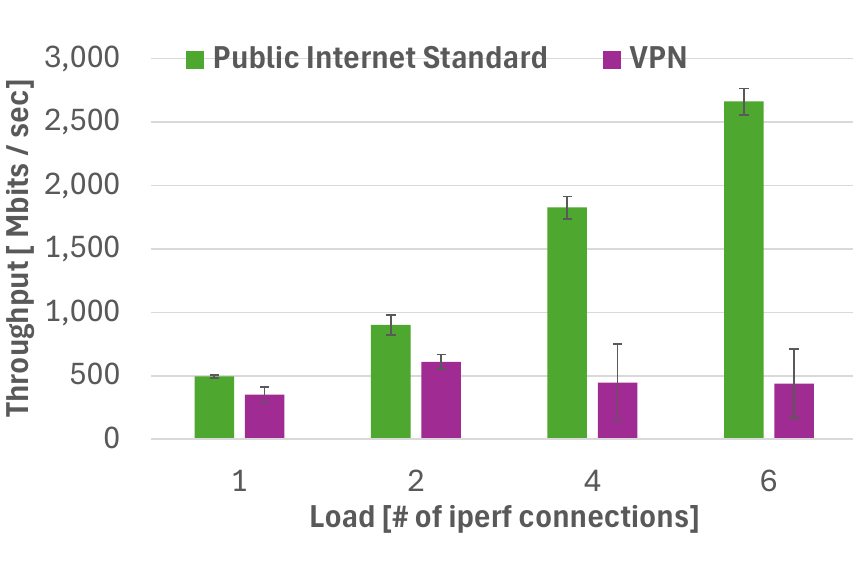}
    \caption{AWS Spain $\rightarrow$ GCP Poland}
\end{subfigure}
\hfill
\begin{subfigure}[t]{.32\linewidth}%
    \includegraphics[page=3,scale=\sscale,trim=2cm 8.8cm 11.5cm 3cm,clip]{figures/vpn_pub.pdf}
    \caption{AWS Spain $\rightarrow$ GCP VA, USA}
\end{subfigure}
\\
\begin{subfigure}[t]{.32\linewidth}%
    \includegraphics[page=5,scale=\sscale,trim=2cm 8.8cm 11.5cm 3cm,clip]{figures/vpn_pub.pdf}
    \caption{GCP Spain $\rightarrow$ AWS Spain}
\end{subfigure}
\hfill
\begin{subfigure}[t]{.32\linewidth}%
    \includegraphics[page=4,scale=\sscale,trim=2cm 8.8cm 11.5cm 3cm,clip]{figures/vpn_pub.pdf}
    \caption{GCP Poland $\rightarrow$ AWS Spain}
\end{subfigure}
\hfill
\begin{subfigure}[t]{.32\linewidth}%
    \includegraphics[page=6,scale=\sscale,trim=2cm 8.8cm 11.5cm 3cm,clip]{figures/vpn_pub.pdf}
    \caption{GCP VA, USA $\rightarrow$ AWS Spain}
\end{subfigure}
\caption{Throughput: VPN Service vs. Unencrypted Public Internet Standard}\label{fig:vpn}
\ReduceVSpace
\end{figure*}}

\newcommand{\figlongvpn}{\begin{figure}
    \centering
    \def\sscale{.38}
    \includegraphics[page=1,scale=\sscale,trim=2cm 9.8cm 3.5cm 2cm,clip]{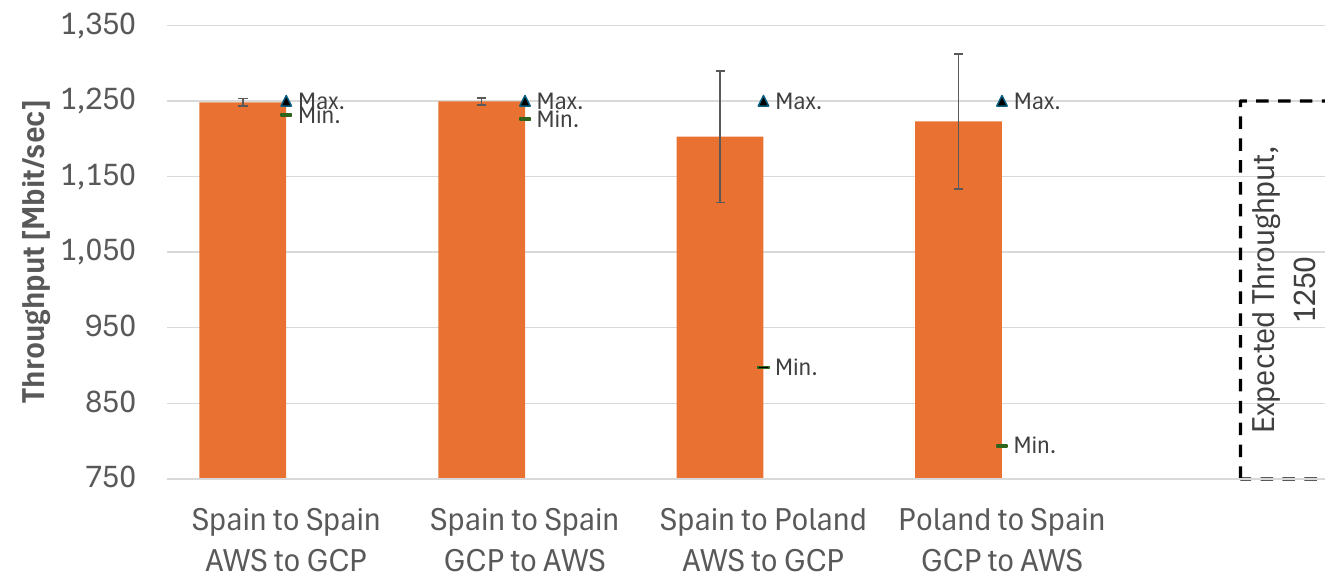}
    \caption{Throughput in long VPN connections
    }
    \label{fig:long-vpn}
    \ReduceVSpace
\end{figure}}

\newcommand{\figcci}{\begin{figure*}
\def\sscale{.38}
\centering
\captionsetup[subfigure]{font=footnotesize,skip=4pt}
\begin{subfigure}[t]{.32\linewidth}%
    \includegraphics[page=2,scale=\sscale,trim=2cm 8.8cm 11.5cm 3cm,clip]{figures/pub_cci.pdf}
    \caption{GCP Spain $\rightarrow$ AWS Spain}
\end{subfigure}
\hfill
\begin{subfigure}[t]{.32\linewidth}%
    \includegraphics[page=1,scale=\sscale,trim=2cm 8.8cm 11.5cm 3cm,clip]{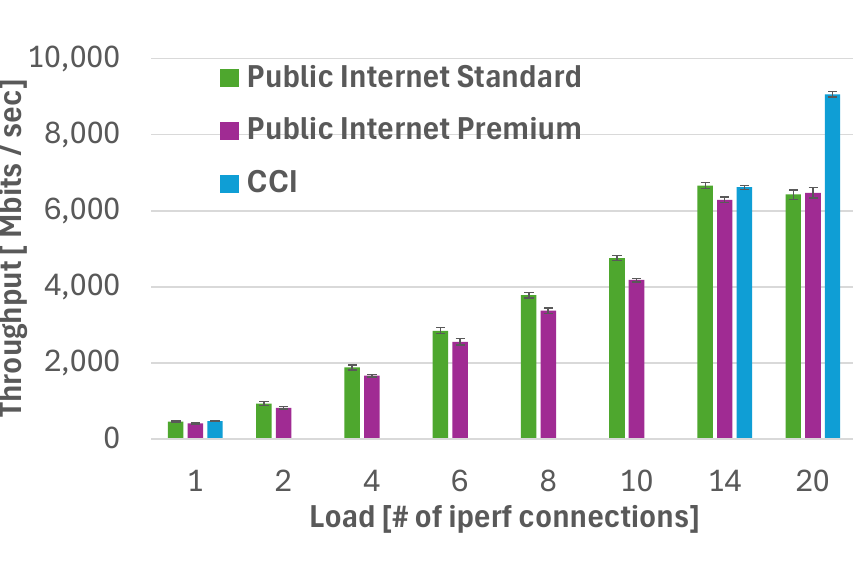}
    \caption{GCP Poland $\rightarrow$ AWS Spain}
\end{subfigure}
\hfill
\begin{subfigure}[t]{.32\linewidth}%
    \includegraphics[page=3,scale=\sscale,trim=2cm 8.8cm 11.5cm 3cm,clip]{figures/pub_cci.pdf}
    \caption{GCP VA, USA $\rightarrow$ AWS Spain}
\end{subfigure}
\caption{Throughput: Public Internet Standard vs. Public Internet Premium vs. CCI
}\label{fig:cci}
\ReduceVSpace
\end{figure*}}

\section{Hands-on evaluation of cross-cloud connectivity}\label{sec:measurements}

\new{In this section, we study whether CCI provides its full capacity in a fair manner, and how it compares to the VPN and Public Internet alternatives, especially when users try to send traffic beyond the nominal capacity.}

\subsection{Preliminary Experiments}\label{subsec:preliminary-experiments}
Prior to performing the main experiments, we need to make sure that VMs do not have bottlenecks on memory or CPU.
We performed experiments with \texttt{iperf3} within the same VPC in the same region in the same cloud for AWS and GCP respectively, testing connectivity between pairs of VM instances of different types. In both GCP and AWS, we were able to obtain the nominal NIC bandwidth guarantees.

\T{Beyond VM NIC nominal capacity.} When considering cross-cloud %
communication, there are caveats that should be taken into account. When ordering a VM instance of a given flavor, it is important to be alert to the fact that this is a virtual, i.e., logical appliance. For example, a VM NIC is, in fact, an elastic cloud resource that shares the underlying physical NIC with other tenants on the same physical host subject to some policy, which is not always fully transparent when cross-cloud communication happens. A NIC of a given nominal capacity can perform from slightly below to considerably higher than its nominal capacity.  For example, \textit{we observed a $2 \times$ higher throughput} on short-lived ($60$ seconds) bursty traffic obtaining  $4.16$ Gbps on a nominally $2$ Gbps NIC. 

Interestingly, this happens across the board for both CCI and Public Internet (premium and regular tiers). It appears that cloud vendors strive to provide the nominal capacity guarantee (even though this is not always achieved on the short-lived connections) and are willing to share spare spot capacity when it is available. Even though this does not entail change in the cost of NIC, the traffic volume is still billed as usual. Thus, this is an advantageous strategy for cloud providers. %
The details of the throttling and capacity sharing policies are not documented though. In general, \textit{our experiments suggest that these distributed mechanisms require some warm-up time} and kick in after some relatively long time, such as $3-5$ minutes. At this point, the observed throughput converges to the nominal capacity. Fully explaining this behavior is out of scope for this paper, but we observed that in a single cloud, convergence to the nominal capacity occurs much faster.

\T{Beyond VLAN nominal capacity.}  
For our evaluation, we used a CCI link of $10$ Gbps and varied the number of VLAN attachments and their capacities between $1-10$ Gbps.
Similarly to NICs, VLAN attachments are not physical VLANs, but logical cloud resources. A VLAN attachment actual performance typically adheres to its nominal capacity, but occasionally for short-lived bursty traffic patterns, it can accommodate much more traffic than allowed by the nominal specification. \textit{We observed up to $70\%$ higher throughput} than the nominal specification. We did not observe VLAN attachment performance below the nominal specification. Note that we explored VLAN attachments and CCI performance using default traffic policies. We were not interested in validating policies for segregating traffic of the same user because they do not affect cost, which is based on the hourly fee per CCI link and  VLAN attachment capacities, and the volume of the egress traffic. 

\T{Beyond CCI nominal capacity.} 
In contrast to the virtual NIC and VLAN attachment resources, we observed that a CCI link never allows to exceed its nominal capacity. This is an expected behavior, since the CCI links are physical resources. 

\T{Overbooking VLANs.} 
In another set of experiments, we want to understand what happens to the traffic of the user under the default policy when  VLAN attachments and/or CCI links are overbooked. It turns out that if a VLAN attachment is overbooked, but the total egress/ingress traffic is below the CCI link's capacity, the TCP connections using this VLAN attachment receive their fair shares. Sometimes the TCP connections can exceed the fair shares, because the total capacity of the VLAN attachment is exceeded thanks to transparent spot capacity sharing by the vendor. Furthermore, when multiple VLANs share the same CCI link, the link capacity is shared fairly between the VLANs, and occasionally VLANs can obtain more than their entitlement. 

We also run heavy overbooking experiments with total VLAN capacity sharing the same CCI link twice exceeding the capacity of the CCI link (i.e., the total VLAN capacity used was $20$ Gbps). In this case, the CCI link allows exactly $10$ Gbps throughput minus $5\%$ L4+L2 overhead, and each of the two $10$ Gbps VLANs receive about $5$ Gbps fair share. 

An additional consideration in cross-cloud connectivity is that \textit{cloud vendors might apply policies to egress and ingress traffic differently}, without making the details completely transparent. All the factors above underline our need to obtain the baseline measurements through experimentation rather than simply assuming that the cloud resources will always perform as their nominal specifications, and use the cheapest VM and VLAN configurations to measure performance of CCI.  
The next section describes the testbed configuration that we settled on following the preliminary experiments.

\figtestbed
\subsection{\new{Testbed and} Methodology}\label{subsec:methodology}
\Cref{fig:testbed} 
depicts our testbed.
In all experiments, VMs run \texttt{Ubuntu 24.04} and we use \texttt{iperf 2.1.9} to measure throughput. We arbitrarily set an \textit{anchor} VPC in AWS region \texttt{eu-south-2} (Madrid, Spain). With respect to this anchor VPC, we use three subnets in the GCP VPC\footnote{Google VPCs are multi-regional global resources comprised of regional private subnets.} in \texttt{europe-southwest1} (Madrid, Spain: intra-region), \texttt{europe-central2} (Poland: intra-continent), and \texttt{us-east4} (Virginia, USA: inter-continent). 
We use \texttt{n2-standard-32} VM type in GCP. It has \texttt{32 VCPUs}, \texttt{128 GB} of memory, and up to \texttt{32 Gbps} default egress bandwidth. In AWS, we use \texttt{m5.12.xlarge} VM type. This machine type has \texttt{48 VCPUs}, \texttt{192 GB} of memory, and \texttt{12 Gbps} of default egress bandwidth.  

For each connectivity option (CCI, VPN, Public Internet Premium, and Public Internet Standard), we measure throughput in two directions, i.e., by using the anchor VPC in AWS as \texttt{iperf3 server} and \texttt{iperf3} client alternately. 

For each experiment configuration (i.e., connectivity option, direction, and collocation option), we also control the utilization level of the cross-cloud link to model (a)~low utilization (below 30\%), (b)~high utilization (around 70\%), and (c)~complete saturation (100\%). The saturation level is being controlled by the number of \texttt{iperf3} connections and bandwidth limit per-connection. 
To illustrate:
\cmd{iperf3 -J -Z -n 20G -i 0 -b 500M \\ -c 10.xxx.xxx.xxx -P 20}
completely saturates the $10$ Gbps CCI connection, albeit for a short period of time. 

We repeat each experiment $30$ times to obtain a data point, compute the average total throughput, and its variance. After a set of $30$ runs, we completely clean the system and start from a clean slate to prevent spurious correlations and side effects across different experiment configurations.  

Thus, overall, we run 4 (connectivity configurations) x 2 (directions) x 3 (collocation options) x 3 (utilization levels) x 30 (replications) $= 2160$ experiments in our main experimentation study. We show a subset of measurements in this paper and will share the full set publicly upon publication.

To make sure that the experiments are not affected by transient conditions in the clouds, we repeat the experiments at different random dates. As one can appreciate, such experiments takes time and budget. They are also difficult to complete in one pass because now and then some exceptions and errors happen. 
Because of the volume of the experiments, the sometimes painstaking discovery process of undocumented behavior of the cloud VPN networking, the configuration changes required as explained in the next subsection, and time and the costs involved, we were not able to run every experiment under exactly the same conditions. To that end, we segregate the data sets obtained under different conditions and discuss them separately.

\figvpn
\figlongvpn
\figcci

\subsection{Cross-cloud  Measurements}\label{subsec:cross-cloud-vpn}

The VPN tunnel bandwidth is limited \emph{nominally} by AWS Site-to-Site VPN quota, which allows $1.25$ Gbps~\cite{AWS-vpn-quotas} per tunnel. The GCP CloudVPN quota allows $3$ Gbps per tunnel. Since we wanted to keep the costs down, we only used one VPN tunnel in the VPN-based experiments and, therefore, we expected that the best throughput attainable in our experiments, will be limited to $1.25$ Gbps. 

\T{Bypassing VPN throttling.}
We discovered that CloudVPN can behave very differently for different traffic profiles. In many cases, we obtained throughput that exceeds the upper limit of the AWS Site-to-Site VPN. With the help of the AWS team, we discovered that \textit{if the traffic flow is short lived, the throttling mechanisms do not kick in fast enough. }

\T{Slow gateway auto-scaling.} Likewise, we observed that in many cases, even for the inbound traffic, we obtained very low throughput as shown in~\cref{fig:vpn}. 
With the help of the AWS team, we discovered that because of some of the elasticity mechanism configurations in AWS, auto-scaling of the gateways kicks in after at least $5$ minutes of sustained high volume traffic going into the AWS cloud from GCP. These features were not documented at the time of our experimentation, and \textit{we consistently obtained sub-optimal results for VPN-based connectivity}, as our experiments lasted less than $5$ minutes. Following the  recommendation of the AWS team, we extended the time of an experiment to be over $5$ mins, as in: 
\cmd{iperf3 -Z -t 600 -i 0 -b 125M \\
-c 10.1.1.6 -P 10}
This allowed to obtain the $1.25$ Gbps upper limit promised by the AWS Site-to-Site VPN quota.

Likewise, when exploring with the AWS customer support team why we do not obtain the upper limit on the AWS outbound traffic, we were able to collaboratively fix this issue on the AWS side. 

\T{Throughput in long VPN connections.}
The results of the long runs are shown in \Cref{fig:long-vpn}. For the intra-region case, the throughput of VPN connectivity is very close to $1.25$ Gbps and the variance is very low. In the inter-region case, the throughput is a bit lower, because of the increased delay, as expected. The variance also becomes larger. It should be stressed that the data points of the measurements were collected at different days and at different hours of the day. In both cases, there are several outliers occasionally crossing the $1.25$ Gbps threshold slightly (as could be expected from our preliminary experiments).

\subsection{CCI Measurements}\label{subsec:cci-measurements}

\cref{fig:cci} shows throughput measurements for CCI and Public Internet (standard and premium tiers). As explained above, we only show CCI performance for $30\%$, $70\%$, and $100\%$ link utilizations. However, we show more data points for the Public Internet options. We make the following observations. 

\T{Guaranteed CCI capacity.} At $100\%$ utilization level, CCI attains its nominal throughput capacity as expected (minus $5\%$ L2+L4 overhead) in both intra-region and intra-continent settings. In the inter-continent configuration, a larger delay causes throughput to drop consistently with the bandwidth-delay product. 

\T{Low CCI utilization.}
At the lower levels of utilization, CCI does not achieve better than standard or premium Internet. %

\T{Internet tier.}
Occasionally, Public Internet standard tier outperforms the premium one. This happens for example in the intra-continent setting when traffic is sent from Poland (GCP) to Madrid (AWS). The premium public Internet tier offering promises to carry traffic inside the GCP network and emit it at the closest Point of Presence (POP) of AWS to hand it to AWS. On the other hand, public standard Internet tier is emitted as soon as possible (i.e., Poland) and enters AWS at the closest POP in Central Europe. The traffic is then carried by the AWS network. Hence, the obtained result might reflect the relative speeds of the cloud vendors and peculiarities of routing.
Of course, this phenomenon does not happen in the intra-region setting, because both vendors have presence in Madrid, Spain, and there are no routing options that would create an asymmetry between the standard and premium traffic. In this case, a user can save costs and receive similar service by using the cheaper standard tier rather than a more expensive premium one.  

In addition, it appears that \textit{the egress Public Internet is capped at $7$ Gbps}, as 
the same NIC is capable of filling the $10$ Gbps CCI link when we use CCI.

\section{Model and Problem Formulation}\label{sec:model}

\new{In the previous section, we analyzed the performance characteristics of CCI. We now want to compare its cost against the VPN alternative. To do so, we }
develop an automated algorithm that dynamically decides when to establish or release either a dedicated physical CCI link or VPN tunnels, in order to minimize the total cost of serving the traffic demands. \new{We start by formulating this as a formal problem.}

\T{Assumptions.} \camera{We now attempt to formulate an optimization problem in the hypothetical case where an oracle knows all demands in advance. We later show that if future demands are unknown, there is no optimal decision, and thus we later offer a heuristic algorithm.} As discussed in Section~\ref{sec:background}, several practical considerations influence the cost model and its estimation:
(1)~The cost of using VPN is not constant; it follows a tiered pricing scheme in which the per-GB cost decreases with the total volume of data transferred.
(2)~Establishing a CCI link involves a provisioning delay of at least  a few business days, denoted by the parameter~\(\Delay\). From our hands-on experience with CCI, we assume $\Delay = 72$ hours. To reflect practical management constraints, we also assume a fixed leasing period of \( \contract =1\)\,week that prevents turning off CCI right away.
In contrast, we assume that a VPN tunnel can be both provisioned and released immediately. 
To highlight the fundamental differences between CCI and VPN, we focus on a simplified scenario in which all regions are located within the same continent and share a single CCI. As a result, the decision reduces to toggling between using VPN and CCI.

\T{Problem framework.}
We consider a set of user pairs $\mathcal{P}$, where each pair $p \in \mathcal{P}$ consists of two VPCs located in different clouds: one in region $r_1$ of provider $A$ and the other in region $r_2$ of provider $B$.  
Our goal is to minimize the total cost over a time horizon of \( T \) hours. %
At each hour, the cost includes two components: an hourly leasing cost and a per-GB data transfer cost. We decompose each cost component into per-pair costs.
\begin{equation}\label{eq:total}
\min \sum_{p \in \mathcal{P}} \sum_{t=1}^T \Big( \text{leasing cost}(p, t) + \text{transfer cost}(p, t) \Big)
\end{equation}

\T{CCI costs. }If $p$ uses CCI at time~\( t \), its leasing cost is \( \frac{\leaseCCI}{P^{t}} + \leaseVlan \), where \( \leaseCCI \) is the shared cost of activating CCI, \( P^{t} \) is the number of pairs sharing CCI at time \( t \), and \( \leaseVlan \) is the per-pair cost for a VLAN attachment. The transfer cost is \( \demand \cdot \costCCI \), where \( \demand \) is the data volume in GB, transferred by $p$ at time \( t\) and \( \costCCI \) is the per-GB transfer cost over CCI.

\T{VPN costs.} If $p$ uses VPN instead, its leasing cost is \( \leaseVPN \), a fixed amount per pair, and the transfer cost is \( \demand \cdot \costVPN \). The VPN transfer cost \( \costVPN \) depends on the accumulated volume from the beginning of the month, so we assume a function \( f(p, \sum_{t'= \text{start of month}}^{t} d_{t'}^p) \) that returns the current per-GB rate for each pair~\( p \) based on its total volume to date.

\T{Problem formulation.}
\Cref{eq:total} becomes
\begin{equation} \label{eq:single_obj}
\begin{aligned}
\min_{\{x_t\}} \sum_{t=1}^T \Bigg[ 
&\ x_t \cdot \left( \leaseCCI
    + \sum_{p \in \mathcal{P}} \left( \leaseVlan 
    + \costCCI \cdot \demand \right) \right) \\
&+ (1 - x_t) \cdot \sum_{p \in \mathcal{P}} \left( \leaseVPN 
    + \costVPN \cdot \demand \right)
\Bigg]
\end{aligned}
\end{equation}

The variable \( x_t \) indicates whether the CCI link is active at time~\( t \). We assume that when CCI is active (\( x_t = 1 \)), all pairs use CCI instead of VPN.

\section{Algorithm}\label{sec:algo}

In this section, we present a \camera{heuristic} algorithm to solve the problem defined in \cref{sec:model} \camera{when future demands are unknown}.

\T{Comparison with ski-rental approach.} The \textit{ski-rental problem} is a foundational online problem that models the trade-off between an ongoing rental cost and a one-time purchase cost, under uncertainty about the future~\cite{karlin1988competitive,karlin1994competitive}. Formally, a skier rents skis at cost $r$ per day or buys them outright for $B$, but the number of skiing days $d$ is unknown. 
The optimal deterministic algorithm rents for $B/r$ days, then buys, achieving a competitive ratio of $2$.
Our setting differs 
in several ways:  
(1)~Leasing CCI (“buying”) is temporary, and the user can switch back to VPN (“renting”). %
(2)~The cost structure includes multiple interacting components (e.g., fixed lease, per-unit transfer, and tiered costs), unlike the simple rent-or-buy setup.  
(3)~There is a built-in provisioning delay \Delay between the decision and activation of the resource.

\begin{figure}[!t] %
\centering

\includegraphics[width=0.9\columnwidth]{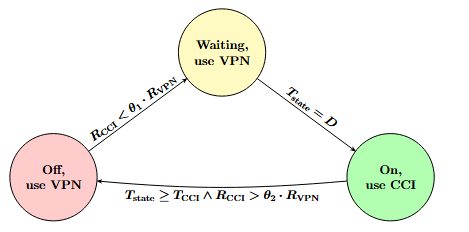}
\caption{ \name's state diagram.}

\label{fig:algoDiagram}
\ReduceVSpace %
\end{figure}

\T{\name.} 
We present an online algorithm, \textit{\name}, that dynamically decides when to activate a dedicated CCI link, based on recent demand and cost trends. It maintains a sliding window of length~\(\History\) to track the aggregated cost of using VPN, denoted \(\aggVPN\), and the cost of using CCI, denoted \(\aggCCI\). Its behavior is governed by a finite-state process, illustrated in~\cref{fig:algoDiagram}. It switches between three states, \textsc{Off}, \textsc{Waiting}, and \textsc{On}, using two thresholds, \(\thresholdEnter\) and \(\thresholdStay\). 
During the initial time steps \( t \in [0, \History) \), \name uses the cumulative cost 
from the past \( t \) steps only. 
For our evaluations, we set \( \contract = \History = 168 \) hours (one week).

In the \textsc{Off} state (which is the initial state), all traffic is routed via VPN. The algorithm remains in this state as long as the observed cost of CCI is not significantly lower than the VPN cost. Once the aggregated CCI cost becomes attractive (specifically, when \(\aggCCI < \thresholdEnter \cdot \aggVPN\)), the algorithm requests CCI provisioning, and transitions into the \textsc{Waiting} state.

While in the \textsc{Waiting} state, traffic continues over VPN during the provisioning delay period. The algorithm remains in this state until the delay \(\Delay\) elapses, after which it transitions to the \textsc{On} state and begins routing traffic over  CCI. 
We denote by \( \tstate \) the time elapsed since \name entered its current state. \begin{red}
In the \textsc{On} state, CCI remains active for at least \(\contract\) time units. After this minimum commitment period, the algorithm continues using CCI as long as it remains cost-effective. Specifically, it stays in the \textsc{On} state while \( \aggCCI \le \thresholdStay \cdot \aggVPN \). Once \( \aggCCI > \thresholdStay \cdot \aggVPN \), CCI is deactivated and the algorithm transitions back to the \textsc{Off} state, reverting to VPN.
\end{red} This approach allows \name to 
toggle between VPN and CCI based on current 
cost dynamics, while incorporating real-world constraints such as provisioning delay. The use of two distinct 
thresholds, \(\thresholdEnter\) for initiating CCI provisioning and \(\thresholdStay\) for renewal, with $\thresholdEnter<\thresholdStay$ ensure stability and prevent oscillation between states, while a single equal threshold can lead to more oscillations, as a well-known hysteresis phenomenon. In the evaluations, we set \(\thresholdEnter = 0.9\) and \(\thresholdStay = 1.1\), reflecting a conservative policy that avoids premature activation of CCI and tolerates mild cost fluctuations.

\T{Theoretical properties.} %
We show that \name is near-optimal at consistently high or low rates.

\begin{property}\label{pro:1}
\noindent \textit{(i)}~If traffic rates remain low, \ie below the activation threshold of \name, then \name achieves optimal cost. %

\noindent \textit{(ii)}~If traffic rates remain high, \ie above the activation condition of \name, 
then \name achieves asymptotic optimal cost. 
\end{property}

\begin{proof}
(i)~If 
traffic rates are sufficiently low, then over any sliding window of length \( \History \), 
$\thresholdEnter \cdot \aggVPN < \aggCCI$, so \name never transitions to the \textsc{Waiting} or \textsc{On} state. 
Since the optimal offline algorithm faces the same cost comparison and finds no benefit in activating CCI, it also remains with VPN. Thus, the behavior and costs are identical.

(ii)~Assume that for every time \(t\),  the aggregated cost over the past \(\History\) steps satisfies
$
\aggCCI < \thresholdEnter \cdot \aggVPN. 
$
This ensures that the algorithm transitions from the \textsc{Off} state to the \textsc{Waiting} state.
After a provisioning delay of \(\Delay\) steps, CCI becomes available. If traffic remains high, the condition
$
\aggCCI < \thresholdStay \cdot \aggVPN
$
is also satisfied in all future windows, ensuring that the algorithm remains in the \textsc{On} state indefinitely, just like the optimal offline solution.
However, due to the provisioning delay \(\Delay\) and the initial time required until the condition for \(\thresholdEnter\) is met, the algorithm continues to use VPN. Thus, the only cost difference between the algorithm and the offline optimum arises during this transition period.
Let

{\small
\[
\gamma = \sum_{t=0}^{\History + \Delay - 1}  \sum_{p \in \mathcal{P}} \left( \demand \cdot (\costVPN - \costCCI) + \leaseVPN - \frac{\leaseCCI}{|\mathcal{P}|} - \leaseVlan \right).
\]
}
\(\gamma\) represents the extra cost incurred by \name due to using VPN while waiting for the CCI provisioning to complete.
After this delay, both the online and offline algorithms use CCI continuously, so their costs become identical. %
Therefore, the competitive ratio satisfies $\frac{\text{Cost}_{\text{\name}}}{\text{Cost}_{\text{\OPT}}} \le 1 + \frac{\gamma}{\text{Cost}_{\text{\OPT}}}.$
As traffic continues over a long time horizon, the additive gap \(\gamma\) becomes negligible relative to the total cost, and the algorithm becomes asymptotically optimal.
\end{proof}

\begin{red}

We now show that unlike the famous competitive ratio of 2 in the ski rental problem \cite{karlin1988competitive}, \textit{there can be no constant competitive ratio} that is independent of the problem parameters in our online problem, no matter the algorithm.

\begin{theorem}\label{thm:no_constant_ratio}
For any constant \( \alpha > 0 \), there exists a traffic pattern and a set of cost parameters such that no online algorithm can guarantee a cost within a factor \( \alpha \) of the optimal offline solution.
\end{theorem}
\begin{proof}
Assume by contradiction that there exists an online algorithm \(\mathcal{A}\) and a constant \(\alpha > 0\) such that for any traffic pattern, the cost of \(\mathcal{A}\) is at most \(\alpha\) times the cost of the optimal offline algorithm.
We construct an adversarial one-step scenario. At time \( t = -D \), algorithm \(\mathcal{A}\) must choose between using VPN or activating a CCI link, without knowing future demand. Thus, at time $t=0$, either CCI is activated and we pay for it, or we use VPN.
\begin{itemize}
    \item If \(\mathcal{A}\) picks VPN, the adversary immediately injects at time $t=0$ a large demand volume \(\demand \gg 0\). The optimal choice would have been to use CCI (with a sufficient CCI capacity parameter), incurring %
    cost \(\frac{\leaseCCI}{P^{k}} +\leaseVlan + \costCCI \cdot \demand\), while \(\mathcal{A}\) pays \(\leaseVPN + \costVPN \cdot \demand\). As \(\demand \to \infty\),
    the cost ratio becomes arbitrarily large:
    \[
    \frac{\text{Cost}_{\mathcal{A}}}{\text{Cost}_{\OPT}} = \frac{\leaseVPN + \costVPN \cdot \demand}{\frac{\leaseCCI}{P^{k}} +\leaseVlan + \costCCI \cdot \demand} \to \frac{\costVPN}{\costCCI} > \alpha
    \]
    for suitable \(\demand\), \(\costVPN\), \(\costCCI\).

    \item If \(\mathcal{A}\) chooses to activate CCI, the adversary sends no traffic at $t=0$. The optimal cost is zero, while \(\mathcal{A}\) incurs a cost of at least \(\leaseCCI\), and 
    $
    \frac{\text{Cost}_{\mathcal{A}}}{\text{Cost}_{\OPT}} = \infty > \alpha.
    $
\end{itemize}

In both cases, the adversary forces the cost ratio beyond any constant \(\alpha\), contradicting the assumption.
\end{proof}

\end{red}

\section{Evaluation}\label{sec:eval}

\subsection{Methodology}

\T{Costs.} We design an evaluation environment to mimic the cost dynamics of cross-cloud data transfers between \camera{AWS,  GCP and Azure}. 
We use official pricing data from AWS~\cite{aws_pricing, aws_direct_connect_pricing}, GCP~\cite{gcp__premium_pricing_catalog,CCI_pricing_catalog} \camera{and Azure~\cite{azure_expressroute_pricing,azure_VPN_pricing}} to get all tiered pricing components. 
Our model reflects the structure and pricing logic of real-world deployments, including transfer fees, VPN overheads,  CCI leasing models, \camera{and volume-based cost thresholds}. 
Although a full production setup includes additional components, such as load balancers and redundancy configurations, we focus here on the main components that distinguish between CCI and VPN. %

\T{Traffic.} Each simulation scenario comprises multiple region pairs, where each pair represents a directional traffic flow between a region in GCP and a region in AWS \camera{or Azure}, or vice versa. For this study, all regions are randomly selected within those of a single 
continent, defined as either Europe or the US. 
Traffic can be routed either through VPN or via CCI. Routing decisions are made hourly\camera{, as a higher frequency would not gain much given the 72-hour CCI provisioning delay}. 
Traffic traces are either synthetic or derived from real-world measurements, as detailed below.

\begin{red}
\T{Algorithms.} 
We compare the \name\ against four baselines: 
(1)~\AlwaysCCI, 
(2)~\AlwaysVPN, 
(3)~\AvgAll, which makes decisions based on the average traffic over the entire history, and 
(4)~\AvgMonth, which makes decisions based on the average traffic over the last month.

\T{Workloads.}
There are no publicly-available workloads that specifically characterize cross-cloud transfers. Therefore, we leverage datasets from related domains and adapt them to our setting. We use four workloads to evaluate our algorithm.
Two use real-world datasets \camera{with different traffic characteristics}:   \textsc{MIRAGE-2019} \cite{aceto2019mirage}\camera{, with bursty mobile app traffic,} and \textsc{PUFFER} \cite{PufferSite}\camera{, with more stable video traffic}.
The other two are synthetic: one with a constant traffic rate, and one that models bursty traffic volumes. 
\end{red}

\subsection{MIRAGE real-world dataset}

\T{Description.}
The MIRAGE-2019 dataset~\cite{aceto2019mirage} was collected at the University of Napoli between 2017 and 2019. It includes traffic logs from over 280 participants using rooted Android devices under realistic mobile usage. %

\T{Pre-processing. }
The raw dataset provides second-by-second traffic traces.
In the preprocessing phase, we aggregate the original traffic data into hourly intervals, producing a continuous 2-year trace. We model a network with $K$ users, each representing an individual mobile device. For each user, we generate a traffic trace by sampling from the MIRAGE dataset: Each day, we randomly select one of the available device traces and assign its hourly traffic volume to that user. 
We assume that each application server is deployed in a random distinct region within GCP, while users are randomly spread across same continent regions within AWS \camera{or Azure---and vice versa}.

\begin{figure*}[!t] %
\centering
\begin{subfigure}{0.48\columnwidth}
  \includegraphics[width=\textwidth]{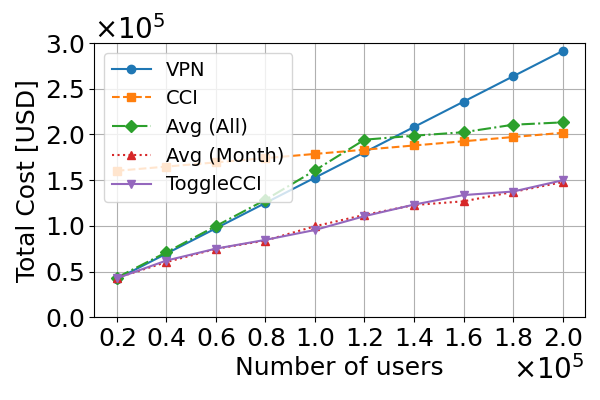}
  \caption{EU: GCP $\to$ AWS}
  \label{fig:mirage:mirage_vs_number_users}
\end{subfigure}
\hfill
\begin{subfigure}{0.48\columnwidth}
  \includegraphics[width=\textwidth]{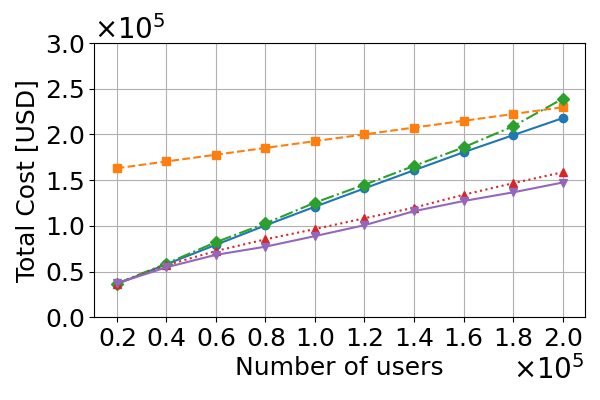}
  \caption{EU: AWS $\to$ GCP}
  \label{fig:mirage:mirage_vs_number_users_aws}
\end{subfigure}
\hfill
\begin{subfigure}{0.48\columnwidth}
  \includegraphics[width=\textwidth]{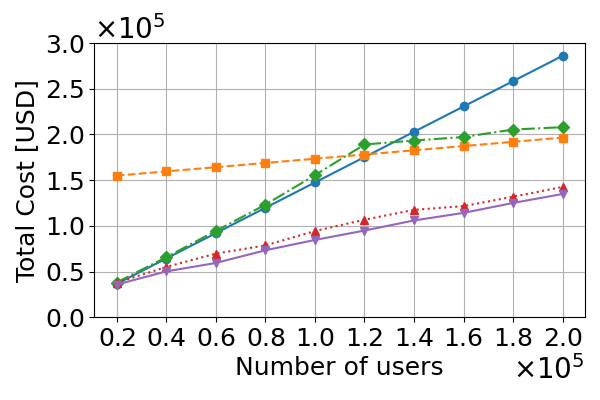}
  \caption{USA: GCP $\to$ AWS}
  \label{fig:mirage:mirage_vs_number_users_usa}
\end{subfigure}
\hfill
\begin{subfigure}{0.48\columnwidth}
  \includegraphics[width=\textwidth]{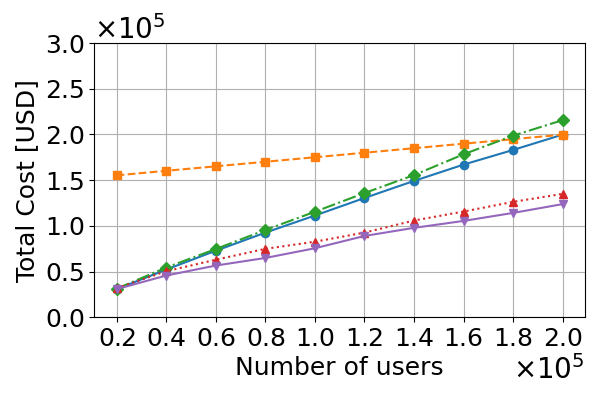}
    \caption{USA: AWS $\to$ GCP}
  \label{fig:mirage:mirage_vs_number_users_aws_usa}
\end{subfigure}
\caption{Performance comparison using the MIRAGE dataset. Cost as a function of the number of users, both in Europe and in the US, from GCP to AWS and vice versa. In all cases, \name is close to optimal.
}

\label{fig:mirage_results}
\ReduceVSpace %
\end{figure*}

\begin{figure*}[!t] %
\centering
\begin{subfigure}{0.48\columnwidth}
  \includegraphics[width=\textwidth]{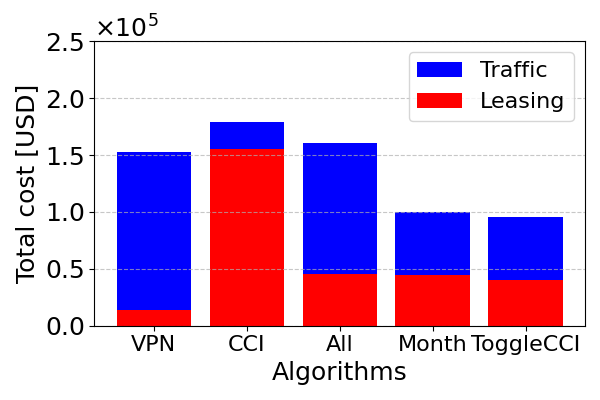}
   \caption{EU: GCP $\to$ AWS}
  \label{fig:mirage:bars_eu_gcp}
\end{subfigure}
\hfill
\begin{subfigure}{0.48\columnwidth}
  \includegraphics[width=\textwidth]{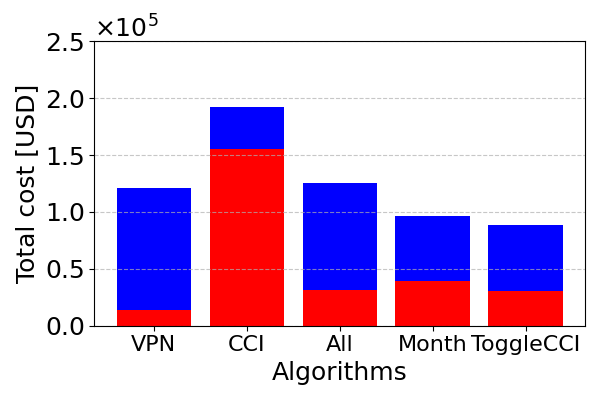}
  \caption{EU: AWS $\to$ GCP}
  \label{fig:mirage:bars_eu_aws}
\end{subfigure}
\begin{subfigure}{0.48\columnwidth}
  \includegraphics[width=\textwidth]{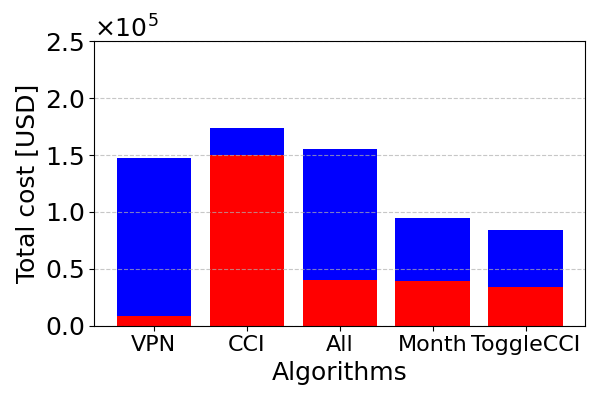}
   \caption{USA: GCP $\to$ AWS}
  \label{fig:mirage:bars_US_gcp}
\end{subfigure}
\hfill
\begin{subfigure}{0.48\columnwidth}
  \includegraphics[width=\textwidth]{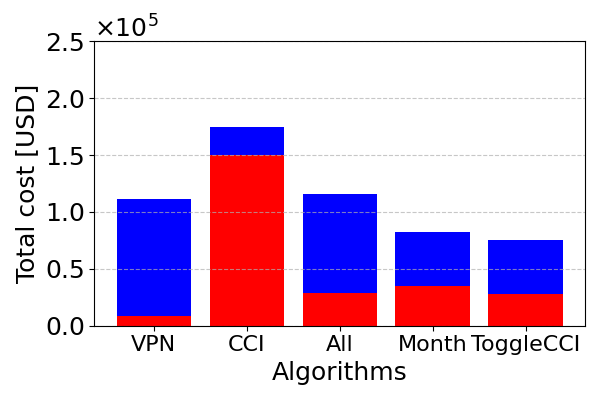}
  \caption{USA: AWS $\to$ GCP}
  \label{fig:mirage:bars_US_aws}
\end{subfigure}
\caption{Cost breakdown for a specific scenario with 100{,}000 users on the MIRAGE dataset, in Europe and in the US, from GCP to AWS and vice versa. Results are roughly similar in all scenarios. \name uses a balanced cost split.
}
\label{fig:mirage_bars}
\ReduceVSpace %
\end{figure*}

\T{Results.}
\Cref{fig:mirage_results} shows the total cost as a function of the number of users. %
Its results are similar to those on the synthetic traffic.
\name consistently tracks the cheapest strategy across varying traffic phases, avoiding unnecessary leasing costs during low-traffic intervals and selectively activating CCI contracts during sustained bursts. ~\Cref{fig:mirage:mirage_vs_number_users,fig:mirage:mirage_vs_number_users_aws,fig:mirage:mirage_vs_number_users_usa,fig:mirage:mirage_vs_number_users_aws_usa} present all four evaluation settings: traffic sent from GCP to AWS and vice versa, in both Europe and the USA. While there are minor variations due to regional pricing and traffic direction, the overall trend is consistent across all scenarios. %
\name reduces the total cost by an average factor of \(1.8\times\) compared to both strategies at breakeven traffic rates.

\Cref{fig:mirage_bars} decomposes the total cost for a representative case (\(K = 100,000\) users) into leasing and traffic components, showing that \name achieves a balanced cost split.

\begin{red}

\T{GCP-Azure.} \Cref{fig:mirage_results_AZURE} presents results for transfers between GCP and Azure in both directions, to evaluate \name's robustness.
The principles of the cost model are similar to those in the GCP-AWS setting, with only minor differences in pricing~\cite{azure_expressroute_pricing}. 
The results are similar to the GCP-AWS case. 
This is because, although the absolute costs differ, the underlying principles remain the same.

\begin{figure}
\centering
\begin{subfigure}{0.48\columnwidth}
  \includegraphics[width=\textwidth]{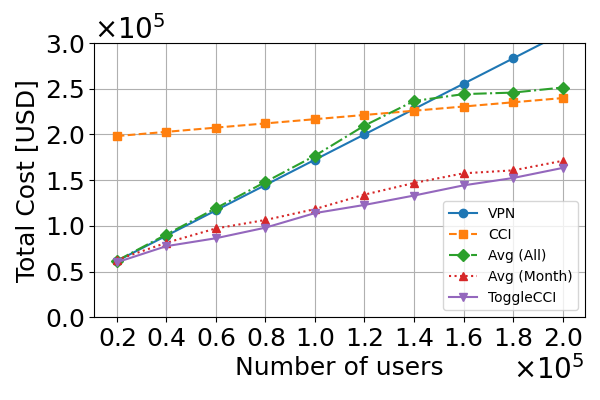}
  \caption{EU: GCP $\to$ AZURE}
  \label{fig:mirage:mirage_vs_number_users_gcp_azure}
\end{subfigure}
\hfill
\begin{subfigure}{0.48\columnwidth}
  \includegraphics[width=\textwidth]{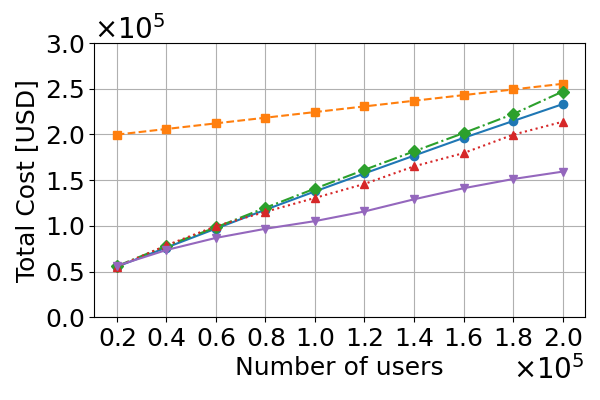}
\caption{EU: AZURE $\to$ GCP}
\label{fig:mirage:mirage_vs_number_users_azure_gcp}
\end{subfigure}
\caption{Performance comparison using the MIRAGE dataset. Cost as a function of the number of users, from GCP to AZURE and vice versa.}%
\label{fig:mirage_results_AZURE}
\ReduceVSpace %
\end{figure}

\T{Inter-Continental Scenario.}
In many practical settings, organizations operate multi-cloud resources across different continents, where inter-continental data transfer costs may be substantial.
We consider a simple broadcast use case: a sender located in Europe (i.e., Paris) sends data from GCP to AWS regions distributed across both Europe and the United States. We assume that the colocation facility is given in this setting, and the user must decide whether to send the traffic through CCI at this colocation or through a VPN. 
We compare two configurations depending on the location of the colocation facility: (1) near the sender (in Paris), and (2) far from the sender (in Ohio). Of course, the egress cost of CCI differs significantly depending on the location of the colocation facility. When the facility is not located near the sender, traffic must first traverse the cloud provider’s inter-continental backbone before reaching the CCI, thereby increasing the overall cost.

\Cref{fig:mirage_inter_continental} shows that in both cases, whether the colocation is nearby or distant, \name adapts to price changes and remains cost-effective in inter-continental settings.

\begin{figure}[!t] %
\centering
\begin{subfigure}{0.48\columnwidth}
  \includegraphics[width=\textwidth]{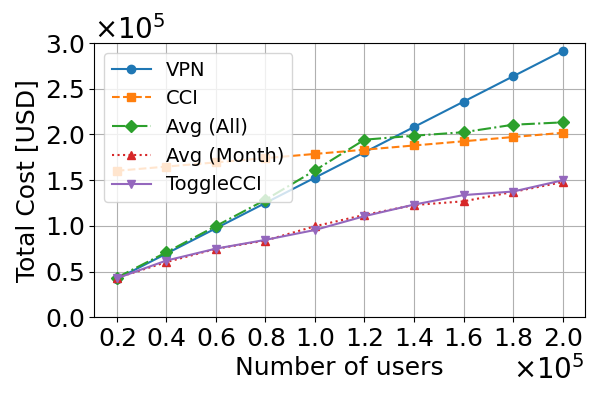}
  \caption{Colocation in Europe}
  \label{fig:inter_continent_colo_in_eu}
\end{subfigure}
\hfill
\begin{subfigure}{0.48\columnwidth}
  \includegraphics[width=\textwidth]{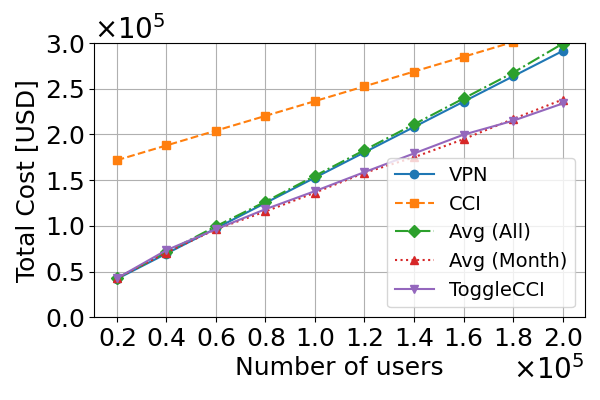}
  \caption{Colocation in the US}
  \label{fig:inter_continent_colo_in_us}
\end{subfigure}

\caption{Inter-continental data transfer cost comparison under different CCI colocation placements (Europe vs. US).}

\label{fig:mirage_inter_continental}
\ReduceVSpace %
\end{figure}
\end{red}

\subsection{PUFFER real-world dataset}

\T{Description.}
The PUFFER dataset \cite{pufferPapper} was collected as part of a real-world video streaming experiment led by Stanford University. Its goal was to study how users interact with adaptive bitrate (ABR) video delivery over the internet. The dataset contains fine-grained, second-level traffic traces from users watching live and on-demand video through a custom ABR system.
The PUFFER dataset provides realistic, continuous traffic patterns generated by users streaming video content. Compared to the bursty mobile app traffic of MIRAGE, it represents more stable, session-based load.

\T{Pre-processing.}
Similarly to our processing of the MIRAGE dataset, we aggregate the data into hourly bins to match the time resolution of our evaluation.  
The dataset includes seven video channels, each assigned to a distinct region in Europe, with transfers directed from GCP to AWS.

\T{Results.}
The PUFFER dataset is characterized by stable, session-based traffic with observable daily and weekly cycles. However, due to the overhead and setup time required to establish a CCI connection, \name cannot effectively exploit these periodic patterns. As a result, dynamic switching between transfer modes provides limited benefit, and the optimal choice between VPN and CCI primarily depends on the overall traffic volume.

\Cref{fig:puffer_a} shows how \name effectively identifies the more cost-efficient option under these conditions. CCI emerges as the most efficient strategy for the PUFFER workload, and \name quickly aligns with it.
In \cref{fig:puffer_b}, we show the decomposition of total cost into leasing and traffic components. As expected, CCI dominates in leasing cost, while VPN dominates in traffic cost. In this scenario, \name closely tracks the cost behavior of CCI.

\begin{figure}[!t] %
\centering
\begin{subfigure}{0.48\columnwidth}
    \includegraphics[width=\textwidth]{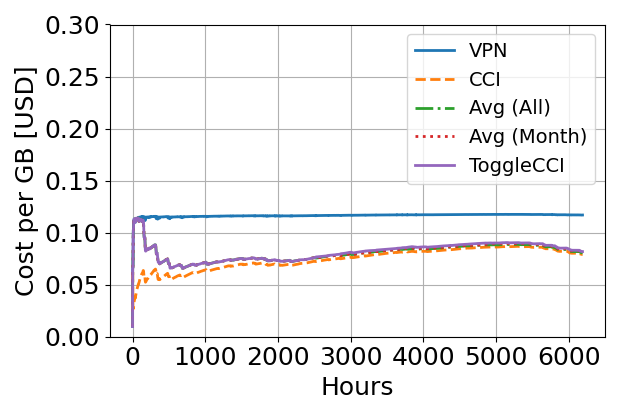}
    \caption{Cumulative average cost per GB}
    \label{fig:puffer_a}
\end{subfigure}
\hfill
\begin{subfigure}{0.48\columnwidth}
    \includegraphics[width=\textwidth]{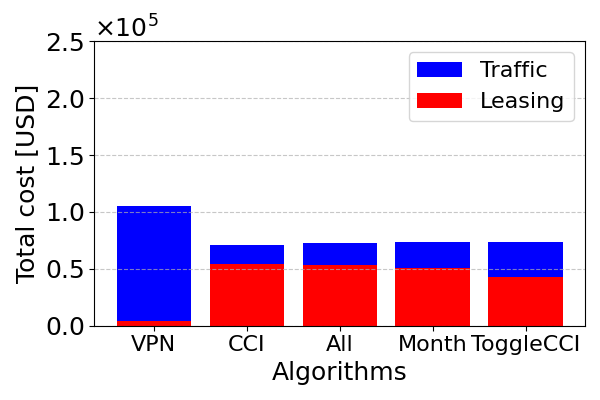}
    \caption{Total cost}
    \label{fig:puffer_b}
\end{subfigure}
\caption{Evaluation with the PUFFER dataset. \name tends to stick with CCI due to the high traffic volume.}
\label{fig:puffer_results}
\ReduceVSpace
\end{figure}

\subsection{Sensitivity Analysis}
\begin{figure*}
\begin{minipage}[t]{0.24\textwidth}    
    \includegraphics[width=\columnwidth]{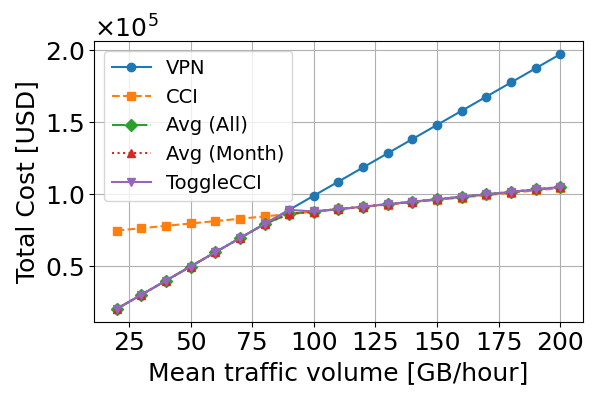}
    \vspace{\parsep}
    \caption{Total cost for the trace with constant traffic rate as a function of rate. \name is near-optimal.
    }
    \label{fig:synthetic_results_a}
\end{minipage}
\hfill
\begin{minipage}[t]{0.74\textwidth}    
\centering
\begin{subfigure}{0.32\linewidth}
    \includegraphics[width=\textwidth]{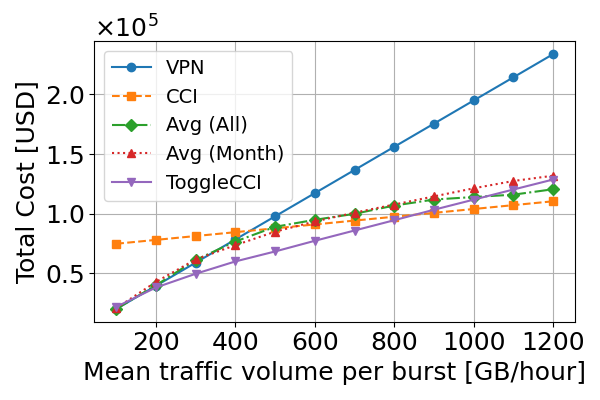}
    \caption{Total cost}
    \label{fig:bursty_total_cost_}
\end{subfigure}
\hfill
\begin{subfigure}{0.32\linewidth}
    \includegraphics[width=\textwidth]{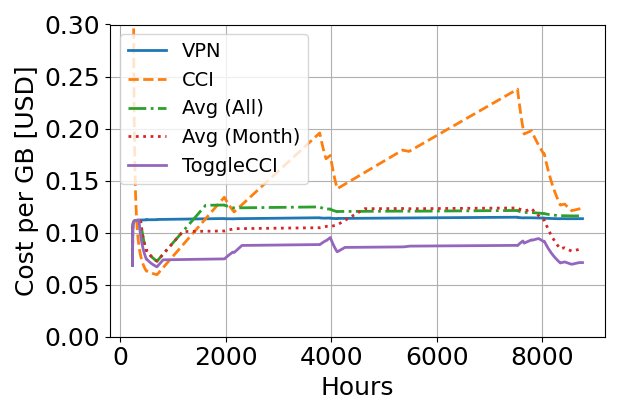}
    \caption{Cumulative average cost per GB %
    }
    \label{fig:bursty_eff_cost}
\end{subfigure}
\hfill
\begin{subfigure}{0.32\linewidth}
    \includegraphics[width=\textwidth]{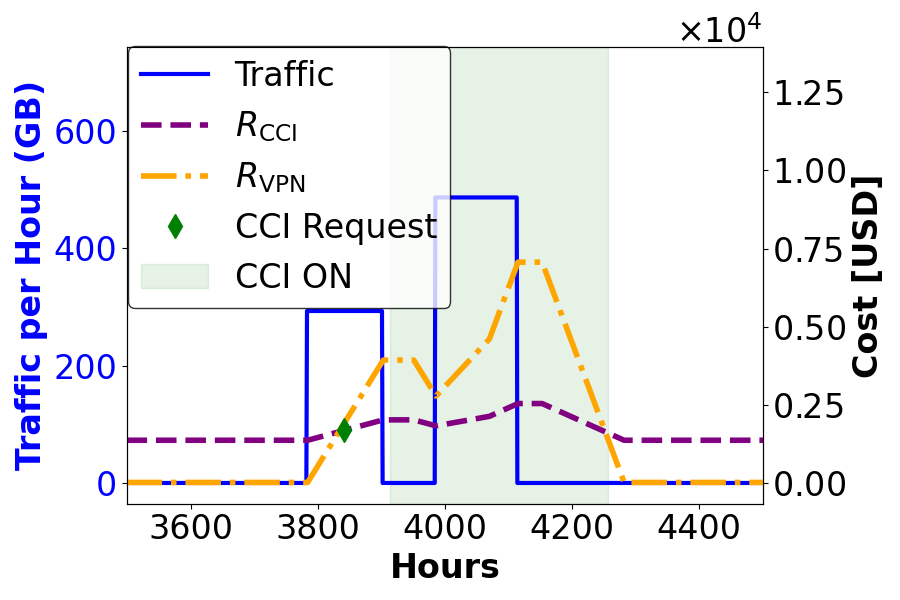}
    \caption{Dynamic behavior}
    \label{fig:bursty_dynamic}
\end{subfigure}
 \vspace{-1mm} %
\caption{Bursty trace. (a)~compares total costs across different mean traffic volumes. Given a mean volume of 400~GB/hour, (b)~shows \name's ability to achieve the best cost per GB compared to the baselines, while (c)~visualizes \name's dynamic behavior over time, with a zoom-in 
between time 3500 and 4500.}
\label{fig:bursty_results}
\end{minipage}
\end{figure*}
\T{Constant traffic.} For additional insights and sensitivity analysis, we also evaluate \name on synthetic workloads. The first is a constant-rate workload, generated over a full year using hourly time steps (\ie 8{,}760 hours). This trace provides initial intuition by fixing the traffic volume at every time step. \camera{It corresponds to scenarios involving short recurring transfer cycles (\eg hourly or daily batches for backups), which appear almost constant to \name.}%

\Cref{fig:synthetic_results_a} presents results for this constant trace. For low traffic rates, it is optimal to stick to VPN, while for higher rates, CCI is preferable.
For low and high rates, \name quickly adopts the optimal choice and
achieves the lowest possible cost, as stated in~\cref{pro:1}. When using CCI, it only misses the first $\Delay$ days due to the CCI setup time. Moreover, because the switching threshold is slightly below the breakeven point (\(\thresholdEnter = 0.9\)), \name remains conservative just below the breakeven point when CCI and VPN costs are equal.

\T{Bursty traffic.} We also evaluate a \textit{bursty trace} for one year. It captures a more dynamic behavior. 
Burst arrivals follow a Poisson process. Burst durations and intensities are sampled from Gaussian distributions with configurable means. %
We start with an initial configuration that uses a Poisson arrival rate of \( \lambda = 1/730 \), corresponding to roughly one burst per month on average; a mean burst duration of about one week; and an average traffic intensity  of 400~GB/hour. %
All results are averaged over 20 randomized repeats.

\Cref{fig:bursty_results}(a)
shows the results for different values of mean burst intensity.
As for the constant trace, for small traffic volumes, VPN is cheaper. For high volumes, CCI becomes the most cost-effective option. In the intermediate range, \name  outperforms both static strategies by dynamically adapting to the observed traffic patterns. It selectively activates CCI during high-demand burst periods and falls back to VPN when demand drops. %
\Cref{fig:bursty_results}(b) shows the average cumulative cost per GB over time given a mean volume of 400 GB/hour. The VPN cost is near-constant, %
while CCI cost drops during bursts and rises in silent periods. \name achieves the lowest cost by adaptively switching between the two.
\Cref{fig:bursty_results}(c) illustrates the behavior of \name over time. We observe that during the idle period (3500-3800), there is no traffic and thus $\aggVPN=0$, while $\aggCCI$ remains stable around 1500\$, which corresponds to the cumulative CCI leasing cost over the sliding window $\History$.
Once the burst begins,  $\aggVPN$ increases sharply due to rising traffic, whereas $\aggCCI$ grows only moderately. \name does not activate CCI until the condition \(\aggCCI<\thresholdEnter\cdot\aggVPN\) is satisfied. The actual activation occurs after a provisioning delay of \(\Delay\) hours. %
After activation, CCI remains \textsc{On} for $\contract$ hours, and is renewed if the stay condition still holds. Otherwise, CCI is deactivated.

\begin{figure}[!t] %
\centering
\begin{subfigure}{0.48\columnwidth}
    \includegraphics[width=\textwidth]{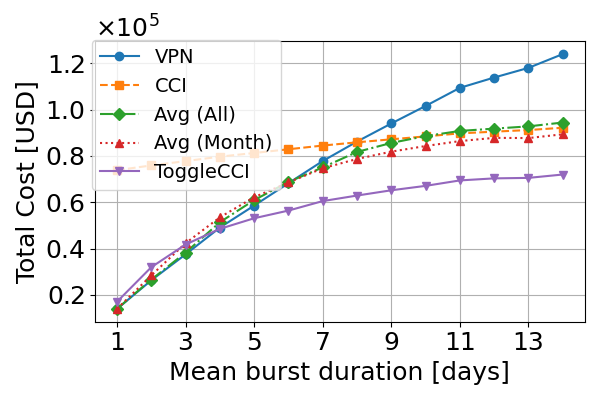}
    \caption{Cost \vs burst duration}
    \label{fig:bursty_duration}
\end{subfigure}
\hfill
\begin{subfigure}{0.48\columnwidth}
    \includegraphics[width=\textwidth]{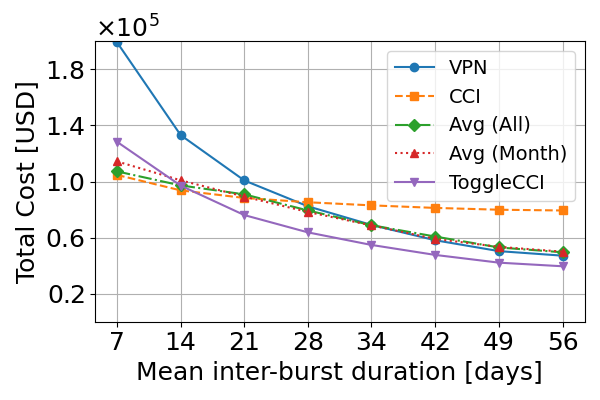}
    \caption{Cost \vs inter-burst time $\frac{1}{\lambda}$}
    \label{fig:burst_lambda}
\end{subfigure}
\caption{Sensitivity analysis for the bursty traffic parameters. (a)~varies the burst duration, given one burst per month on average. (b)~varies the inter-burst interval, given a burst duration of 7 days  on average.%
}
\label{fig:bursty_tuning_params}
\ReduceVSpace %
\end{figure}
\Cref{fig:bursty_tuning_params} presents a sensitivity analysis for bursty traffic. 
In \cref{fig:bursty_duration}, we vary the burst duration while keeping the average frequency fixed at one burst per month. When the duration is low, \name is more expensive than VPN due to the fixed contract duration time $\contract$. %
However, when the duration exceeds $\Delay$ days, \name becomes more cost-effective than VPN. %
In \cref{fig:burst_lambda}, we fix the average burst duration to 7 days and vary the interval between bursts. In this scenario, \name outperforms VPN for any value of $\lambda$ because the duration is sufficiently long. When the interval between bursts becomes short, CCI performs the best due to the limitation of $\Delay$. Once the time between bursts becomes large (around 21 days), \name outperforms both.

\camera{
We next examine the robustness of \name by varying the provisioning delay $\Delay$.
\Cref{fig:sensitivity_delay} presents a high-traffic scenario, where CCI is cheaper than VPN. In this regime, \name outperforms the static policies when the provisioning delay is short, since the algorithm can respond quickly to traffic changes and activate CCI in time. However, when the delay becomes large, \name reacts slowly, and \AlwaysCCI becomes a better option.
\Cref{fig:sensitivity_delay_breakeven} shows the same experiment for the breakeven scenario, where VPN and CCI have similar costs. In this case, \name remains beneficial even under relatively long provisioning delays.}

\begin{figure}[!t] %
\centering
\begin{subfigure}{0.48\columnwidth}
    \includegraphics[width=\textwidth]{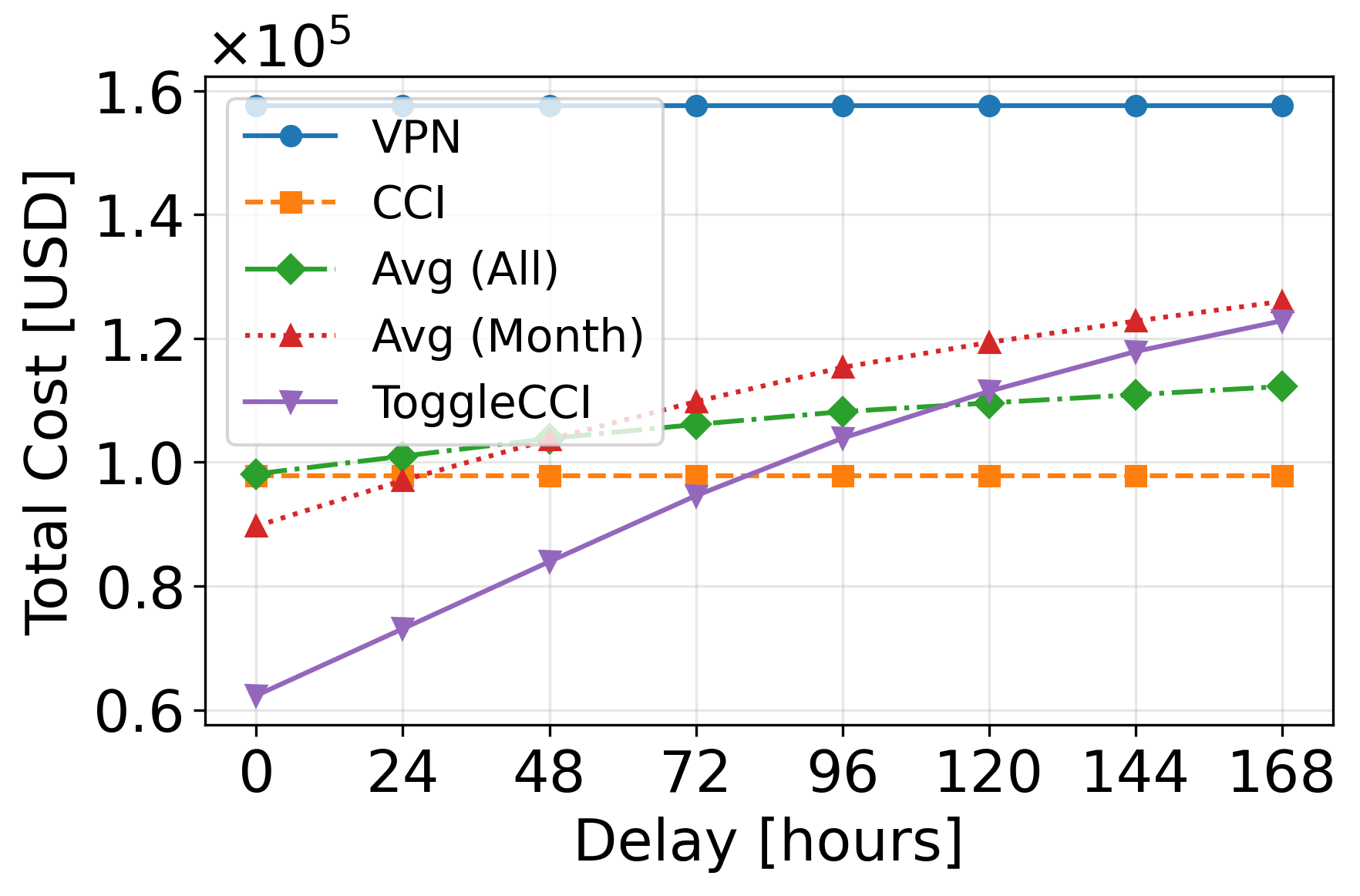}
    \caption{High traffic.}
    \label{fig:sensitivity_delay}
\end{subfigure}
\hfill
\begin{subfigure}{0.48\columnwidth}
    \includegraphics[width=\textwidth]{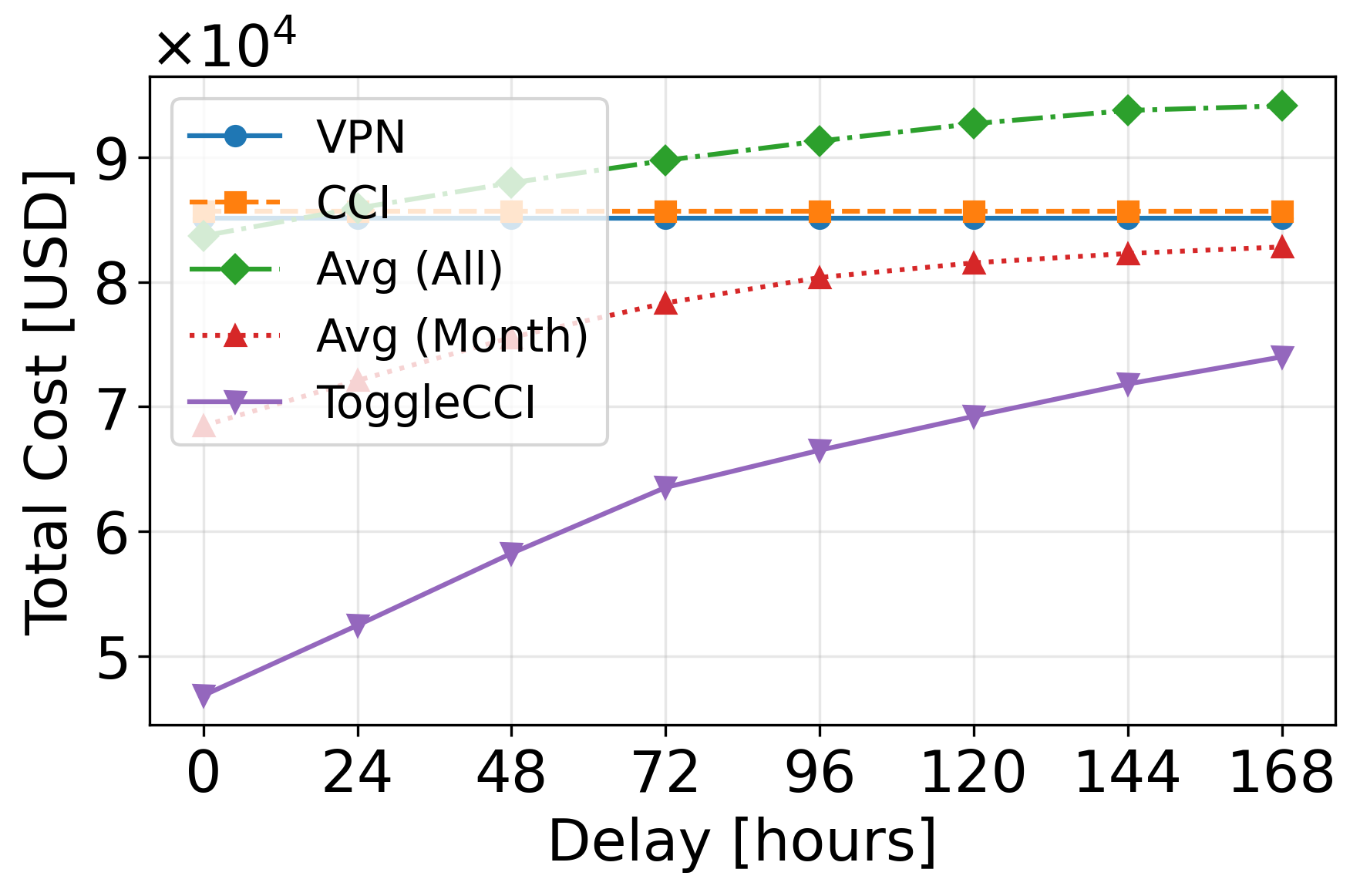}
    \caption{Breakeven traffic.}
    \label{fig:sensitivity_delay_breakeven}
\end{subfigure}
\caption{\camera{Sensitivity of \name to the provisioning delay $\Delay$. The plots show the total cost as a function of the provisioning delay under (a)~high traffic and (b)~mid-range traffic near the VPN/CCI breakeven point.}}
\label{fig:delay_sensitivity}
\ReduceVSpace %
\end{figure}

\section{Conclusion}\label{sec:conclusion}

We presented the first comprehensive study of CCI, a high-throughput option for inter-cloud transfers. While CCI offers low per-GB costs, it uses fixed leasing fees and takes a few days to establish. We conducted a methodical study of CCI performance and compared it to other connectivity options. To the best of our knowledge these results are first of a kind and they provide some valuable insights into the actual (also undocumented) behavior of cross-cloud connectivity.
We \camera{further} proposed \name, an online algorithm that dynamically switches between VPN and CCI based on recent cost trends. We proved that \name achieves near-optimal performance in persistent high or low demand regimes. Evaluations on synthetic and real-world traces show that \name closely tracks the best static choice and significantly reduces total cost in practice. 

\camera{Several directions remain for future work. 
In particular, while we evaluate a single VPN tunnel as a representative baseline, real deployments may use more advanced networking solutions, with more complex cost considerations.}

\T{Acknowledgments.} This work was partly supported by the Louis and Miriam Benjamin Chair in Computer-Communication Networks.

\ifacm %
  
    \ifacmart 

      \bibliographystyle{ACM-Reference-Format}
      \bibliography{mybib}
    \else
      \bibliographystyle{abbrv}
      \bibliography{mybib}    
    \fi

\else
	\ifusenix %
    	{\footnotesize 
        \bibliographystyle{acm}
        \bibliography{mybib}
        }
    \else
        \ifhotnets %
            \bibliographystyle{abbrv} 
            \begin{small}
                \bibliography{mybib}
            \end{small}
          \else %
            \bibliographystyle{IEEEtran}
            \bibliography{mybib}%
        \fi %
    \fi %
\fi %

\end{document}